\documentclass[12pt]{article}
\usepackage{amsmath}
\usepackage{amsthm}
\usepackage{graphicx,psfrag,epsf}
\usepackage{amsfonts}
\usepackage{float}
\usepackage{enumitem}
\usepackage{natbib}
\usepackage[nottoc,numbib]{tocbibind}
\usepackage{booktabs,caption,subcaption}
\usepackage[flushleft]{threeparttable}
\usepackage{multirow}
\newcommand{\ra}[1]{\renewcommand{\arraystretch}{#1}}

\usepackage{xr}

\newcommand{\blind}{1}

\newcommand{\bA}{\boldsymbol{A}}

\newcommand{\bj}{\boldsymbol{j}}
\newcommand{\bS}{\boldsymbol{S}}
\newcommand{\bs}{\boldsymbol{s}}

\newcommand{\bV}{\boldsymbol{V}}
\newcommand{\bv}{\boldsymbol{v}}

\newcommand{\bX}{\boldsymbol{X}}

\newcommand{\bY}{\boldsymbol{Y}}
\newcommand{\by}{\boldsymbol{y}}
\newcommand{\bmu}{\boldsymbol{\mu}}
\newcommand{\bbeta}{\boldsymbol{\beta}}
\newcommand{\bGamma}{\boldsymbol{\Gamma}}
\newcommand{\bgamma}{\boldsymbol{\gamma}}
\newcommand{\bPsi}{\boldsymbol{\Psi}}
\newcommand{\bpsi}{\boldsymbol{\psi}}
\newcommand{\bG}{\boldsymbol{G}}
\newcommand{\bg}{\boldsymbol{g}}

\newcommand{\bSigma}{\boldsymbol{\Sigma}}

\newtheorem{theorem}{Theorem}

\newtheorem{lemma}{Lemma}
\theoremstyle{definition}

\allowdisplaybreaks

\begin{document}

\def\spacingset#1{\renewcommand{\baselinestretch}%
{#1}\small\normalsize} \spacingset{1}


\if1\blind
{
  \title{\bf A Distributed and Integrated Method of Moments for High-Dimensional Correlated Data Analysis}
  \author{Emily C. Hector\thanks{
    This research was funded by grants NSF DMS1513595, NIH R01ES024732 and NIH P01ES022844.}\hspace{.2cm}\\
    Department of Biostatistics, University of Michigan\\
    Peter X.-K. Song \\
    Department of Biostatistics, University of Michigan}
    \date{}
  \maketitle
} \fi

\if0\blind
{
  \bigskip
  \bigskip
  \bigskip
  \begin{center}
    {\LARGE\bf Title}
\end{center}
  \medskip
} \fi

\bigskip
\begin{abstract}
This paper is motivated by a regression analysis of electroencephalography (EEG) neuroimaging data with high-dimensional correlated responses with multi-level nested correlations. We develop a divide-and-conquer procedure implemented in a fully distributed and parallelized computational scheme for statistical estimation and inference of regression parameters. Despite significant efforts in the literature, the computational bottleneck associated with high-dimensional likelihoods prevents the scalability of existing methods. The proposed method addresses this challenge by dividing responses into subvectors to be analyzed separately and in parallel on a distributed platform using pairwise composite likelihood. Theoretical challenges related to combining results from dependent data are overcome in a statistically efficient way using a meta-estimator derived from Hansen's generalized method of moments. We provide a rigorous theoretical framework for efficient estimation, inference, and goodness-of-fit tests. We develop an R package for ease of implementation. We illustrate our method's performance with simulations and the analysis of the EEG data, and find that iron deficiency is significantly associated with two auditory recognition memory related potentials in the left parietal-occipital region of the brain.
\end{abstract}

\noindent%
\vfill

\newpage
\spacingset{1.45} 
\section{INTRODUCTION}
\label{sec:intro}

This paper focuses on developing a systematic divide-and-conquer procedure, readily implemented in a parallel and scalable computational scheme, for statistical estimation and inference. We consider a regression setting with high-dimensional correlated responses with multi-level nested correlations. The proposed Distributed and Integrated Method of Moments (DIMM) is flexible, fast, and statistically efficient, and reduces computing time in two ways: (i) in the distributed step, composite likelihood is executed in parallel at a number of distributed computing nodes, and (ii) at the integrated step, an efficient one-step meta-estimator is derived from \cite{Hansen}'s seminal generalized method of moments (GMM) with no need to load the entire data on a common server.\\
Let $\bY_i$ be the $M$-dimensional correlated response for subject $i$, $i=1, \ldots, N$, and $\bmu_i=E(\bY_i \lvert \bX_i, \bbeta)$ the mean response-covariate relationship of interest for some $M\times p$ dimensional matrix of covariates $\bX_i$ and a $p$-dimensional parameter of interest $\bbeta$. We model $\bmu_i$ by a generalized linear model of the form $g(\bmu_i)=\bX_i \bbeta$, where $g$ is a known link function. The difficulties associated with current methods for high-dimensional correlated response modelling stem from computational burdens and modelling challenges associated with a high-dimensional likelihood. The generalized estimating equation (GEE) proposed by \cite{Liang-Zeger}, one of the widely used methods for the analysis of correlated response data, uses a quasilikelihood approach based on the first two moments of the response to avoid the specification of a parametric joint distribution. GEE is not well suited to high-dimensionality due to the potentially large number of nuisance parameters to estimate and the inversion of large matrices; see \cite{Cressie-Johannesson} and \cite{Banerjee-Gelfand-Finley-Sang}. Additionally, common assumptions by GEE on the correlation structure of the response are too simple to capture multi-level nested correlations, resulting in a substantial loss of efficiency; see \cite{Fitzmaurice-Laird-Rotnitzky}. Simple cases where the estimator of the nuisance parameter does not exist are also outlined in \cite{Crowder}.\\
Composite likelihood (CL) was proposed by \cite{Lindsay} as a method to perform inference on $\bbeta$ by only considering low dimensional marginals of the joint distribution. Pairwise CL, in particular, constructs a pseudolikelihood by multiplying the likelihood objects of pairs of observations. In this way, CL is free of the computational burden of inverting high-dimensional correlation matrices associated with GEE and benefits from an objective function that facilitates model selection. Pairwise CL has been used in longitudinal (\cite{Kuk-Nott}, \cite{Kong-Wang-Gray}), spatial (\cite{Heagerty-Lele}, \cite{Arbia}), spatiotemporal (\cite{Bai-Song-Raghunathan}, \cite{Bevilacqua-Gaetan-Mateu-Porcu}), and genetic (\cite{Larribe-Fearnhead}) data analyses. A well-known bottleneck of CL is the high computational cost of evaluating a large number of low-dimensional likelihoods and their derivatives, a problem that is exacerbated with large $M$.\\
The use of CL relies on knowledge of low-dimensional dependencies among $\bY_i$ in order to specify pairwise CLs properly. Fortunately, in practice, observations within $\bY_i$ are often known to belong to homogeneously-correlated groups, or sub-responses, established by previous science: for example, genomic response data can be grouped by gene or genetic function, metabolomic data by pathway, spatial data by proximity, and brain imaging data by brain function regions. This substantive scientific knowledge can be used to strategically partition response variables in order to speed up computations. The method of divide-and-conquer is a state of the art approach to analyzing data that can be partitioned. In the current literature, this method proposes to randomly split subjects into independent groups in the ``divide'' step (or ``Mapper'') and combines results in the ``conquer'' step (or ``Reducer''); see for example kernel ridge regression (\cite{Zhang-Duchi-Wainwright}) and matrix factorization (\cite{Mackey-Talwalkar-Jordan}). The independent groups can be analyzed in parallel, greatly reducing computation time. Extending the divide-and-conquer approach to our problem, we propose to split the high-dimensional correlated response into subvectors to form correlated response groups according to substantive scientific knowledge. Each subvector is analyzed separately, then results from these analyses are combined. While this method is computationally appealing, our groups of data are correlated, leading to new methodological challenges. In particular, correlation between groups of data must be taken into account when combining results. To our knowledge, our method is among the first attempts to establish a rigorous theoretical framework for combining results from correlated groups of data. The key technique to establish the related theoretical framework relies on an extended version of the confidence distribution based on pairwise CL to derive a GMM estimator of $\bbeta$. For discussion on the confidence distribution and related work with independent cross-sectional data, see \cite{Xie-Singh}, \cite{Singh-Xie-Strawderman} and \cite{Liu-Liu-Xie}; for a similar divide-and-conquer approach with independent scalar responses, see \cite{Lin-Xi}. We propose an optimal weighting matrix that non-parametrically accounts for between-group correlations. The resulting DIMM alleviates the computational burden and modelling challenges associated with existing methods. \\
\begin{figure}
\centering
\begin{subfigure}{0.52\linewidth}
\centering
\includegraphics[width=0.95\linewidth]{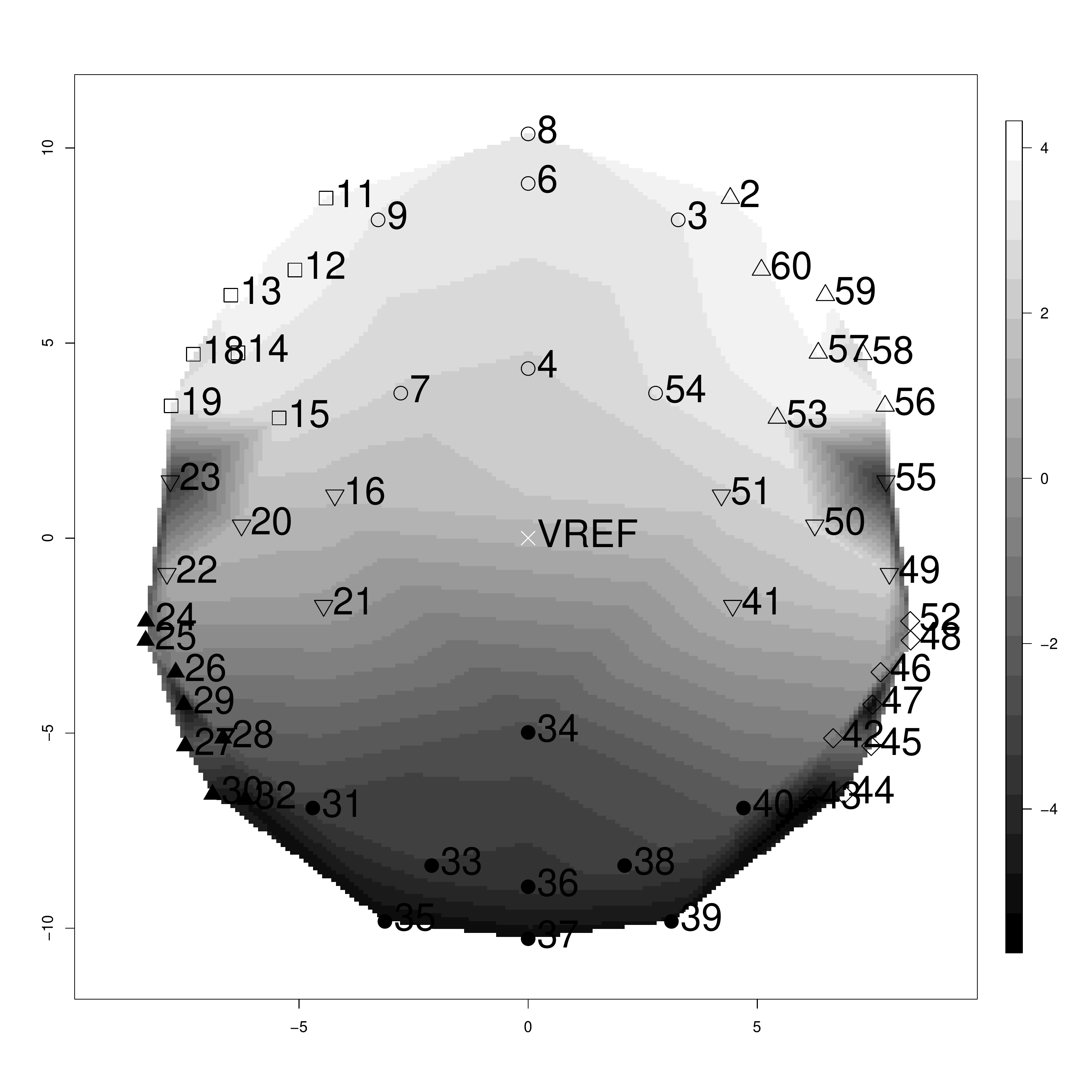}
\caption{}
\label{EEG_plot}
\end{subfigure}%
\begin{subfigure}{0.5\linewidth}
\centering
\includegraphics[width=0.95\linewidth]{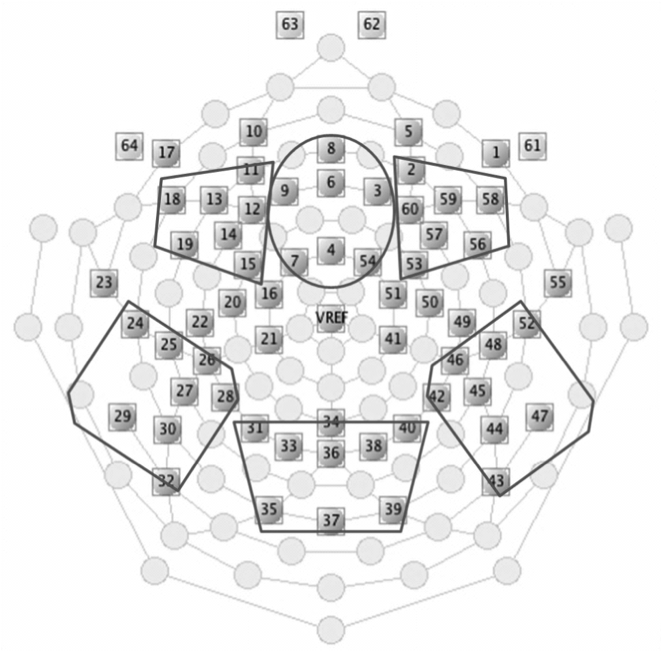}
\caption{}
\label{map-nodes-regions} 
\end{subfigure}
\caption{(a) Average P2 amplitude for iron sufficient children under stimulus of mother's voice. Color plot and additional plots in Supplemental Material. (b) Layout of the 64 channel sensor net with brain regions related to auditory recognition memory.}
\end{figure}
We illustrate our method with a motivating cohort study to assess the association between iron deficiency and auditory recognition memory in infants. Electrical activity in the brain during a 2000 milliseconds period was measured in 157 infants under two vocal stimuli using an electroencephalography (EEG) net consisting of 64-channel sensors on the scalp as visualized in Figure \ref{EEG_plot}. For each sensor and each stimulus, three important event-related potentials (ERPs) related to auditory recognition memory were calculated; as shown in Figure \ref{density-plot}, P2 averages electrical signal between 175 and 300 milliseconds, P750 between 350 and 600 milliseconds, and late slow wave (LSW) between 850 and 1100 milliseconds. The investigator wanted to analyze the data in sub-regions, where $46$ of the nodes belong to six brain function regions related to auditory recognition memory, as seen in Figure \ref{map-nodes-regions}. The complex data-generating mechanism results in a response of dimension $M=46(nodes)\times 3(ERPs)\times 2(stimuli)=276$ that has a multi-level nested correlation structure that is difficult to model, including longitudinal correlations between the three ERP's, spatial correlations between the 46 nodes and within the six brain function regions, and correlations within each voice stimulus. Due to this complex correlation structure and the large number of response variables, traditional methods for correlated data analysis are greatly challenged. \cite{Zhou-Song} developed a method to analyze the LSW outcome, but no existing method is suitable to analyze this dataset in its entirety. We develop DIMM, a fast and efficient method to analyze all 276 responses simultaneously by partitioning the response according to ERPs and brain function regions. DIMM also performs well with higher dimensional correlated outcomes, as seen in simulations.\\
\begin{figure}
\centerline{\includegraphics[width=0.75\linewidth]{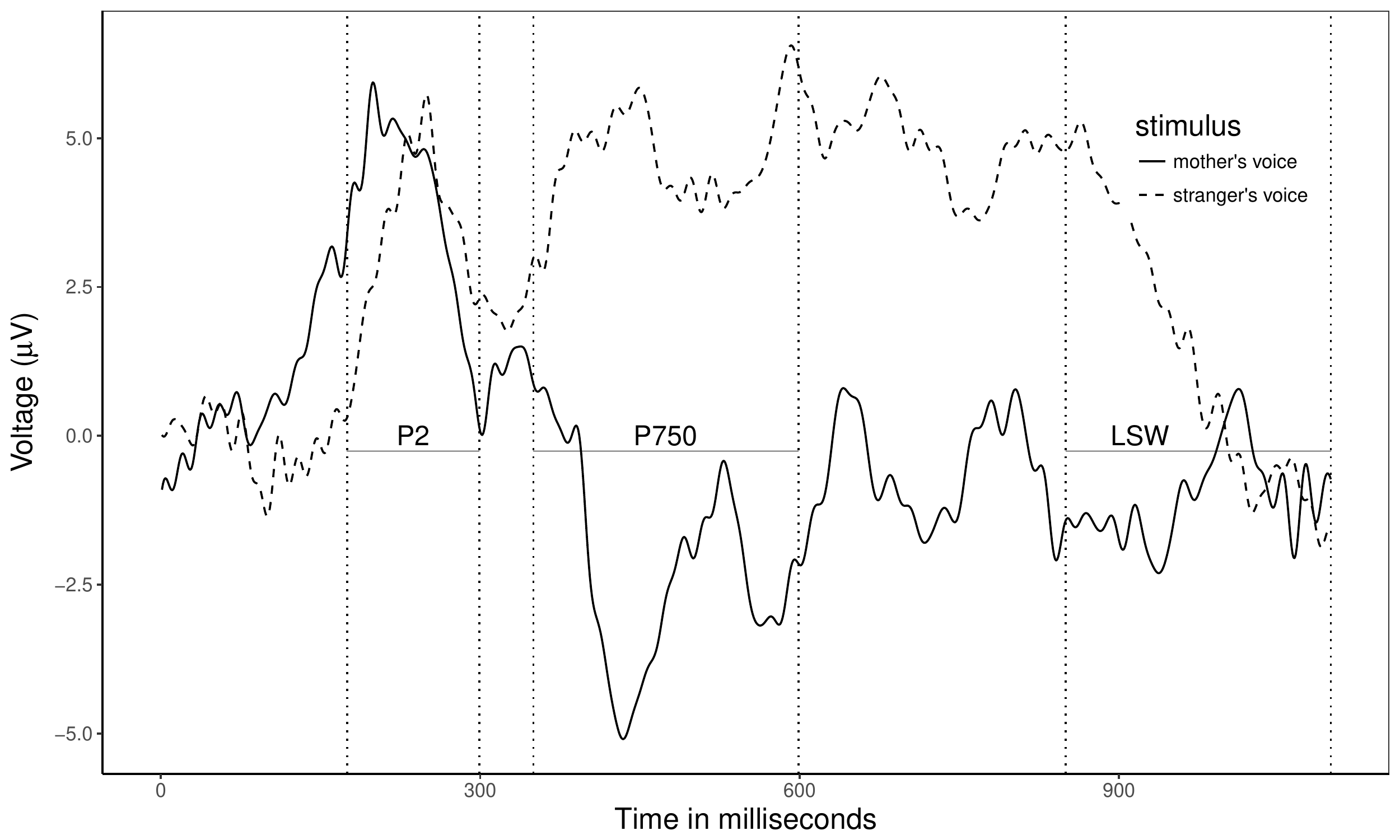}}
\caption{Plot of electrical potential for subject 1 at electrode 2 over time.}\label{density-plot}
\end{figure}
In this paper, we develop a Distributed and Integrated Method of Moments (DIMM) that splits the high-dimensional response into dependent groups of response subvectors according to substantive science, analyzes these smaller groups of responses separately using pairwise CL, and combines the results using an optimal GMM. Our proposed method loses very little estimation efficiency for two reasons: (i) CL performs well on smaller groups of responses with simple but well-approximated local correlation structure; and (ii) we use an optimal weighting matrix in the GMM. More importantly, our method is computationally attractive for two reasons: (i) pairwise CL only evaluates low-dimensional likelihoods and CL analyses can be run in parallel; and (ii) we provide a closed-form of the combined estimator that only depends on CL estimates and group-specific sufficient statistics. Finally, this paper contributes substantially to the existing literature with two key innovations: (i) a rigorous theoretical framework for combining estimates from dependent groups of data that is (ii) scalable to very large $M$. In addition, the proposed method is illustrated on a complex dataset that could previously not be analyzed in its entirety.\\
The rest of the paper is organized as follows. Section \ref{sec:methods} describes DIMM. Section \ref{sec:asymptotics} discusses large sample properties. Section \ref{sec:implementation} presents the closed form one-step meta-estimator, and its implementation in a parallel and scalable computational scheme. Section \ref{sec:simulations} illustrates DIMM's finite sample performance with simulations under the linear model. Section \ref{sec:data} presents the EEG data analysis. Section \ref{sec:discussion} concludes with a discussion. Proofs of theorems and additional simulation and data analysis results are deferred to the Appendix and Supplemental Material.

\section{FORMULATION}
\label{sec:methods}

Let $\left\{ \by_i, \bX_i \right\}_{i=1}^N$ be $N$ independent observations, where the dimension $M$ of $\by_i$ is so big that a direct analysis of the data is computationally intensive or prohibitive. Let $f(\bY_i;\bGamma_i, \bX_i)$ be the $M$-variate joint distribution of $\bY_i \lvert \bX_i$, where $\bGamma_i$ contains parameters of high-order dependencies that may be difficult to handle computationally. We aim to obtain a statistically efficient (small variance) and computationally fast estimator for the regression coefficient $\bbeta$ given the challenges arising from the high-dimensionality and complex dependencies of the response. Our DIMM solution uses a divide-and-conquer approach based on pairwise CL methodology for locally homogeneous data blocks.

\subsection{Division: distributed composite likelihoods}
\label{subsec:divide}

For each $i \in \left\{ 1, \ldots, N \right\}$, we propose to split the $M$-dimensional response $\by_i$ and associated covariates into $J$ blocks $\left\{ \by_{i,j}, \bX_{i,j} \right\}_{i=1}^N$ for $j=1, \ldots, J$, as follows:
\begin{align*}
\by_i&=\left(\begin{array}{ccc} \by_{i,1}^T & \ldots & \by_{i,J}^T \end{array} \right)^T \mbox{ and }\bX_i=\left( \begin{array}{ccc} \bX_{i,1}^T & \ldots & \bX_{i,J}^T \end{array} \right)^T.
\end{align*}
Within block $j$, let $m_j$ be the dimension of subject $i$'s response, $\sum \limits_{j=1}^J m_j=M$, where $\by_{i,j}=\left(y_{i1,j}, \ldots, y_{im_j,j}\right)^T \in \mathbb{R}^{m_j}$ is subject $i$'s $j$th sub-response and $\bX_{i,j} \in \mathbb{R}^{m_j \times p}$ is the associated covariate matrix. Then $\left\{ \by_{i,j} \right\}_{i=1}^N$ are independent realizations of the random variables $\bY_{i,j} \lvert \bX_{i,j}$ whose $m_j$-variate distributions conditional on $\bX_{i,j}$ are denoted by $f(\by_{i,j}; \bGamma_{i,j}, \bX_{i,j})$. Parameter $\bGamma_{i,j}$ encodes information on the marginal moments of $\bY_{i,j}$. This yields $J$ regression models $g_j(\bmu_{i,j})=\bX_{i,j} \bbeta_j$, where $\bmu_{i,j}=E(\bY_{i,j} \lvert \bX_{i,j}, \bbeta_j)$ is the marginal mean of $\bY_{i,j}$, $j=1, \ldots, J$. In most cases, homogeneity of $g_j$ and $\bbeta_j$ holds such that $g_j=g$ and $\bbeta_j=\bbeta$ for $j=1, \ldots, J$; we drop the subscript $j$ by using $\bbeta$ and $g$ to denote $\bbeta_j$ and $g_j$. On some occasions, homogeneity may not hold, for example when each sub-response $\bY_{i,j}$ corresponds to continuous, count, or dichotomous outcomes; in this case, we propose to perform a sub-group analysis by combining regression parameter estimates over blocks where homogeneity in $g_j$ and $\bbeta_j$ holds. In the analysis of the EEG data, we choose $g_j=g$ to be the identity link and perform a sub-group analysis. We suggest splitting the response data according to substantive scientific knowledge, resulting in homogeneous correlations within each response subvector that are suitable for simplifications in structure. If such knowledge is lacking, data pre-processing may help to learn structural features of dependencies. In the linear model, estimation is robust against correlation misspecification; that is, estimators are still consistent but may not be efficient if the data split is not aligned with the true dependence structure.\\
We can obtain an estimate of $\bbeta$ for each of the $J$ blocks of data using pairwise CL methods. The above partition enables us to reduce the challenge of modelling $M$-order dependencies to that of modelling $m_j$-order dependencies of (approximately) local homogeneity. Modelling and estimating high-order moments is practically difficult with more complex tensor data structures, and may be even harder without adequate sample size. In addition, there may be tremendous computational burdens associated with the log likelihood or its derivative, such as the computation of a high-dimensional inverse covariance matrix in the multivariate normal model. To resolve this difficulty, CL has been suggested by many researchers (see \cite{Varin-Reid-Firth} and the references therein) as the method of choice, and takes the following form:
\begin{align}
\mathcal{L}_j(\bbeta, \bgamma_j; \by_{i,j})&=\prod \limits_{r=1}^{m_j-1} \prod \limits_{t=r+1}^{m_j} f_{j}(y_{ir,j},y_{it,j};\bbeta, \bgamma_j, \bX_{i,j}),
\label{PCL}
\end{align}
where $\bgamma_j$ only contains information on second-order moments of $\bY_{i,j}$. The nature of the data partition allows for different dependence parameters $\bgamma_j$, allowing us to make simplifying assumptions on the high-order dependencies of $\bY_{i,j}$. Here, $f_j$ can be chosen according to the data type under investigation as bivariate margins of an $m_j$-variate joint distribution. For example, $f_j$ can be bivariate Normal for continuous data, or, using bivariate dispersion models generated by Gaussian or vine copulas, can be bivariate Poisson or Bernoulli for count or dichotomous data; see Chapter 6 of \cite{Song} and Chapter 3 of \cite{Joe-2}. We set $f_j$ bivariate Normal for the EEG data. Within block $j$, the log-CL for the first and second moment parameters is
\begin{align*}
c\ell_j(\bbeta, \bgamma_j; \by_j)&=\log \prod \limits_{i=1}^N \mathcal{L}_j(\bbeta, \bgamma_j;\by_{i,j}) = \sum \limits_{i=1}^N \sum \limits_{r=1}^{m_j-1} \sum \limits_{t=r+1}^{m_j} \log f_j (y_{ir,j}, y_{it,j}; \bbeta, \bgamma_j, \bX_{i,j}).
\end{align*}
Define composite score functions $\bpsi_j (\bbeta; \by_{i,j}, \bgamma_j)=\sum \limits_{r=1}^{m_j-1} \sum \limits_{t=r+1}^{m_j} \nabla_{\bbeta} \log f_j (y_{ir,j}; y_{it,j}; \bbeta, \bgamma_j, \bX_{i,j})$ and $\bg_j(\bgamma_j; \by_{i,j}, \bbeta)=\sum \limits_{r=1}^{m_j-1} \sum \limits_{t=r+1}^{m_j} \nabla_{\bgamma_j} \log f_j (y_{ir,j}; y_{it,j}; \bbeta, \bgamma_j, \bX_{i,j})$. The pairwise CL estimating equations for the mean and covariance parameters are, respectively:
\begin{align}
\bPsi_j(\bbeta; \by_j, \bgamma_j)=\frac{1}{N} \sum \limits_{i=1}^N \bpsi_j (\bbeta; \by_{i,j}, \bgamma_j)=0 \label{block-EE-1}  \\
\bG_j(\bgamma_j; \by_j, \bbeta)=\frac{1}{N} \sum \limits_{i=1}^N \bg_j(\bgamma_j; \by_{i,j}, \bbeta)=0 \label{block-EE-2} 
\end{align}
Following \cite{Varin-Reid-Firth}, the maximum composite likelihood estimators (MCLE) of $\bbeta$ and $\bgamma_j$ within block $j$, denoted respectively by $\widehat{\bbeta_j}$ and $\widehat{\bgamma_j}$, are the joint solution to the system of unbiased estimating equations in \eqref{block-EE-1} and \eqref{block-EE-2}.\\
Generally, $\bgamma_j$ is block-specific and unknown, and $\widehat{\bbeta_j}$ depends on $\widehat{\bgamma_j}$. When $\bgamma_j$ is a function of $\bbeta$ only, as in generalized linear models, finding $\widehat{\bbeta_j}$ amounts to profile likelihood estimation.  If $\bgamma_j$ is known or absent, then the above simplifies to finding $\widehat{\bbeta_j}$ as the solution to $\bPsi_j(\bbeta; \by_j, \bgamma_j)=0$. We denote $\widehat{\bbeta}_{MCLE}=\left( \widehat{\bbeta_1}^T, \ldots, \widehat{\bbeta_J}^T \right)^T$ and $\widehat{\bgamma}_{MCLE}=\left( \widehat{\bgamma_1}^T, \ldots, \widehat{\bgamma_J}^T \right)^T$. In some practical studies where interest is in block-specific mean parameters and combined dependence parameters, we can treat $\bbeta$ as a nuisance parameter and $\bgamma_j$ as the parameter of interest by switching the roles of $\bPsi_j$ and $\bG_j$ in the following presentation.

\subsection{Integration: the generalized method of moments}
\label{subsec:combine}
We have successfully obtained $J$ estimates of $\bbeta$ based on $J$ estimating equations \eqref{block-EE-1}. In the integration step, we treat each estimating equation $\bPsi_j(\bbeta; \by_j, \widehat{\bgamma_j})=0$ as a moment condition on $\bbeta$ coming from block $j$, $j=1, \ldots, J$. We would like to derive an estimator $\widehat{\bbeta_c}$ of $\bbeta$ that satisfies all $J$ moment conditions. Unfortunately, there is no unique solution to all $J$ estimating equations because they over-identify our parameter; that is, the dimension of parameter $\bbeta$ is less than $Jp$, the dimension of the equation restrictions on $\bbeta$. To overcome this, we invoke \cite{Hansen}'s seminal GMM to combine the moment conditions that arise from each block. Stack the $J$ estimating equations by defining $\bpsi_N(\bbeta; \by_i)=\left(
\bpsi^T_1 (\bbeta; \by_{i,1}, \widehat{\bgamma_1}),
\ldots,
\bpsi^T_J (\bbeta; \by_{i,J}, \widehat{\bgamma_J}) \right)^T$ and 
\begin{align}
\bPsi_N \left( \bbeta; \by \right) &=\left( \begin{array}{ccc}
\bPsi^T_1(\bbeta; \by_1, \widehat{\bgamma_1}) &
\ldots &
\bPsi^T_J(\bbeta; \by_J, \widehat{\bgamma_J}) 
\end{array} \right)^T =\frac{1}{N} \sum \limits_{i=1}^N \bpsi_N(\bbeta; \by_i).
\label{stacked-EE}
\end{align}
Notice the implicit dependence on $\widehat{\bgamma}_{MCLE}$ in \eqref{stacked-EE}. Since $\bPsi_N(\bbeta; \by)=0$ has no solution, following Hansen's GMM we minimize a quadratic form of $\bPsi_N$ with weight matrix $\widehat{\bV}_{N, \bpsi}$, the $Jp \times Jp$ sample variance-covariance matrix of $\bPsi_N(\bbeta; \by)$ evaluated at the MCLE's:
\begin{align}
\widehat{\bV}_{N, \bpsi}&=\frac{1}{N} \sum \limits_{i=1}^N  \left( \bpsi^T_1(\widehat{\bbeta_1}; \by_{i,1}, \widehat{\bgamma_1}), \ldots, \bpsi^T_J(\widehat{\bbeta_J}; \by_{i,J}, \widehat{\bgamma_J}) \right)^{T~\otimes 2},
\end{align}
where $a^{\otimes 2}=a a^T$ for $a \in \mathbb{R}^{Jp}$. Then define the combined GMM estimator of $\bbeta$ as:
\begin{align}
\widehat{\bbeta_c}&=\arg \min \limits_{\bbeta} \left\{ N \bPsi_N^T(\bbeta; \by) \widehat{\bV}^{-1}_{N, \bpsi} \bPsi_N(\bbeta; \by) \right\}=\arg \min \limits_{\bbeta} Q_N(\bbeta).
\label{def:combined-estimator}
\end{align}
We notice similarities of \eqref{def:combined-estimator} to \cite{Qu-Lindsay-Li} but with a completely different way of constructing moment conditions, and to \cite{Wang-Wang-Song-2012} but with a completely different way of partitioning data and the added generality of allowing between-block correlations. The uniqueness of DIMM stems from combining estimating equations $\bPsi_j$ with GMM instead of combining $\widehat{\bbeta_j}$ or data blocks $\left\{\by_{i,j}, \bX_{i,j} \right\}_{i=1}^N$ directly. This new approach allows us to find a GMM estimator $\widehat{\bbeta_c}$ that benefits from a wealth of established theoretical properties. By using the sample variance $\widehat{\bV}_{N, \bpsi}$ we not only account for between-block correlations but find the optimal GMM estimator in the sense that $\widehat{\bbeta_c}$ has variance at least as small as any other estimator exploiting the same moment conditions, hereafter referred to as ``Hansen optimal''. The dimension of $\bPsi_N$ is also smaller than that of $\bY$, reducing the computational burden associated with handling $\bY$ directly.\\
To better understand our estimator, we can show that $\widehat{\bbeta_c}$ maximizes a density in a manner similar to the classic maximum likelihood estimator (MLE) by deriving the quadratic form in \eqref{def:combined-estimator} using an extended version of the confidence distribution (CD) (or density) (\cite{Fisher-1}, \cite{Fisher-2}, and \cite{Efron}). For more discussion on CD and applications to MLE with independent cross-sectional data, refer to \cite{Xie-Singh}, \cite{Singh-Xie-Strawderman}, and \cite{Liu-Liu-Xie}. So far, very little work has been done on the development of CD for correlated data. $\bPsi_j$ are sufficient statistics for $\bbeta$ within each block and are asymptotically Normally distributed under mild assumptions by the Central Limit Theorem (CLT). Their joint distribution is the distribution of $\bPsi_N$, which is also asymptotically Normal under the same mild assumptions of the CLT. Then as long as $\widehat{\bV}_{N, \bpsi}$ is a consistent estimator of the variance of $\bPsi_N$, $\sqrt{N} \widehat{V}_{N, \bpsi}^{-1/2} \bPsi_N(\bbeta_0; \by)$ asymptotically follows a standard normal distribution, where $\bbeta_0$ is the true value of $\bbeta$. By maximizing the distribution of $\bPsi_N$ as a function of $\bbeta$, we can find an estimator that accounts for correlation between sufficient statistics and is the most likely value to arise from the data. We define the confidence estimating function (CEF) as $H_{\bpsi}(\bbeta_0)= \Phi \left( \sqrt{N} \widehat{\bV}_{N, \bpsi}^{-1/2} \bPsi_N (\bbeta_0; \by) \right)$, where $\Phi(\cdot)$ is the $Jp$-variate standard normal distribution function. Then define the density of the CEF as
\begin{align}
h_{\bpsi}(\bbeta)&=\phi \left( \sqrt{N} \widehat{\bV}_{N, \bpsi}^{-1/2} \bPsi_N(\bbeta; \by) \right) \propto \exp \left\{ -\frac{N}{2} \bPsi_N^T (\bbeta; \by) \widehat{\bV}_{N, \bpsi}^{-1} \bPsi_N(\bbeta; \by) \right\},
\label{CEF-equation}
\end{align}
where $\phi(\cdot)$ is the $Jp$-variate standard normal density.  The CEF density has the advantage over the confidence density of not having a sandwich estimator for the variance, and thus not requiring the computation of a sensitivity matrix. It reflects the joint distribution of the $J$ estimating equations \eqref{block-EE-1}. Maximizing $h_{\bpsi}(\bbeta)$ yields the minimization defined in \eqref{def:combined-estimator}. The formulation in \eqref{CEF-equation} is different from the aggregated estimating equation approach proposed by \cite{Lin-Xi} for independent scalar responses.

\section{ASYMPTOTIC PROPERTIES}
\label{sec:asymptotics}

Let $\bv_{\bpsi}(\bbeta)=E_{\bbeta}\left(\bpsi_N(\bbeta; \by_i) \bpsi^T_N(\bbeta; \by_i) \right)$ be the positive definite variability matrix of $\bPsi_N$. Let $\left[ \bv^{-1}_{\bpsi} (\bbeta) \right]_{i,j}$ be the rows $(i-1)p+1$ to $ip$ and columns $(j-1)p+1$ to $jp$ of matrix $\bv^{-1}_{\bpsi} (\bbeta)$. We assume throughout that $\widehat{\bV}_{N, \bpsi}$ is nonsingular. Let $\left\| \cdot \right\|$ be the Euclidean norm. Let $\bbeta_0$, $\bgamma_{j0}$ the true values of $\bbeta$ and $\bgamma_j$, $j=1, \ldots, J$ respectively. Denote $\underline{\bgamma}=(\bgamma_1, \ldots, \bgamma_J)$, $\underline{\bgamma_0}=(\bgamma_{10}, \ldots, \bgamma_{J0})$. Let the variability and sensitivity matrices for block $j$ respectively be
\begin{align*}
\bv_{j, \bpsi_j} (\bbeta)&=Var_{\bbeta}\left\{ \sqrt{N} \bPsi_j(\bbeta; \by_j, \widehat{\bgamma_j}) \right\}=E_{\bbeta}\left\{ \bpsi_j(\bbeta; \by_{i,j}, \widehat{\bgamma_j}) \bpsi^T_j(\bbeta; \by_{i,j}, \widehat{\bgamma_j}) \right\},\\
\bs_{j, \bpsi_j}(\bbeta)&=-\nabla_{\bbeta} E_{\bbeta} \left\{ \bPsi_j(\bbeta; \by_j, \widehat{\bgamma_j}) \right\}=-\nabla_{\bbeta} E_{\bbeta} \left\{ \bpsi_j(\bbeta; \by_{i,j}, \widehat{\bgamma_j}) \right\}.
\end{align*}
Consistency of the block-specific MCLE's $\widehat{\bbeta_j}$ and $\widehat{\bgamma_j}$ is established in Theorem 3.4 of \cite{Song} when the estimating equation $\bPsi_j$ is unbiased at $(\bbeta_0, \bgamma_{j0})$ and its expectation has a unique zero at $(\bbeta_0, \bgamma_{j0})$. $\bPsi_j$ are unbiased if the bivariate marginals $f_j$ are correctly specified. Robustness of consistency to model misspecification is discussed by \cite{Xu-Reid}. They distinguish between misspecification of the full likelihood but correct specification of the pairwise likelihoods, and misspecification of the CL. In the former case, the MCLE is still consistent, whereas the MLE based on the full likelihood is not. In the latter case, the MCLE converges almost surely to the Kullback-Leibler divergence. Existing model diagnostics can help detect ill-posed model structures on the bivariate marginals.\\
As a GMM estimator, $\widehat{\bbeta_c}$ enjoys several key asymptotic properties for valid statistical inference under mild regularity conditions \ref{cond-weight-matrix}-\ref{cond-norm} listed in the Appendix, including consistency and asymptotic normality. We show in Lemma \ref{lemma:weight-matrix} that $\widehat{\bV}_{N, \bpsi}$ converges to the true variance of the estimating equations.
\begin{lemma}[Optimality]
\label{lemma:weight-matrix}
Under condition \ref{cond-weight-matrix}, $\widehat{\bV}_{N, \bpsi} \stackrel{p}{\rightarrow} \bv_{\bpsi}(\bbeta_0)$ as $N \rightarrow \infty$.
\end{lemma}
The proof of Lemma \ref{lemma:weight-matrix}, given in Appendix \ref{sec:appendix:1}, is straightforward, and makes use of the consistency of the MCLE's and the Central Limit Theorem. Lemma \ref{lemma:weight-matrix} shows our GMM estimator is Hansen optimal because we use a weighting matrix that converges to the true variance of the estimating equations. Asymptotically, $\widehat{\bbeta_c}$ does not lose any efficiency beyond what is lost by using CL. Since the pairwise CL is not a full likelihood, there are no general efficiency results about $\widehat{\bbeta_j}$. In unpublished work, it has been shown that the MCLE loses very little to the full likelihood estimator under several popular dependence structures such as compound symmetry, autoregressive and unstructured, and that the pairwise CL is robust to model misspecification and tends to outperform the full likelihood approach when the dependence structure of the data is misspecified (\cite{Jin}). This means DIMM generally performs well. In Theorems \ref{thm:consist} and \ref{thm:norm}, we show that $\widehat{\bbeta_c}$ is consistent and asymptotically normal under mild moment conditions.
\begin{theorem}[Consistency of $\widehat{\bbeta_c}$]
\label{thm:consist}
Given conditions \ref{cond-weight-matrix} and \ref{cond-consist}, $\widehat{\bbeta_c} \stackrel{p}{\rightarrow} \bbeta_0$ as $N \rightarrow \infty$. 
\end{theorem}
The proof of Theorem \ref{thm:consist}, given in Appendix \ref{sec:appendix:1}, derives from the consistency of the GMM estimator due to \cite{Hansen} and, more generally, to \cite{Newey-McFadden}. 
\begin{theorem}[Asymptotic normality of $\widehat{\bbeta_c}$]
\label{thm:norm}
Given conditions \ref{cond-weight-matrix}, \ref{cond-consist} and \ref{cond-norm}, \\$\sqrt{N} \left( \widehat{\bbeta_c}-\bbeta_0 \right) \stackrel{d}{\rightarrow} \mathcal{N} \left(0, \bj_{\bpsi}^{-1}(\bbeta_0) \right)$ as $N \rightarrow \infty$, where the Godambe information of $\bPsi_N(\bbeta; \by)$ can be rewritten as $\bj_{\bpsi}(\bbeta)=\bs^T_{\bpsi}(\bbeta) \bv^{-1}_{\bpsi}(\bbeta) \bs_{\bpsi}(\bbeta)=\sum \limits_{i,j=1}^J \bs_{i, \bpsi_i}(\bbeta) \allowbreak \left[ \bv^{-1}_{\bpsi} (\bbeta) \right]_{i,j} \allowbreak \bs_{j, \bpsi_j}(\bbeta)$.
\end{theorem}
The proof of Theorem \ref{thm:norm} follows from Theorem 7.2 in \cite{Newey-McFadden} and Theorem \ref{thm:consist}. Theorems \ref{thm:consist} and \ref{thm:norm} do not require the differentiability of $\bPsi_j$ and $Q_N$. Instead, they require the differentiability of their population versions, and that $\bPsi_N$ behave ``nicely'' in a neighbourhood of $\bbeta_0$. These conditions allow us, for example, to do quantile regression. These theoretical results provide a framework for constructing asymptotic confidence intervals and conducting Wald tests, so that we can perform inference for $\bbeta$ when $M$ is very large. Using an optimal weight matrix improves statistical power so DIMM can detect signals other methods may miss.\\
So far, we have been vague about how the data partition should be done, only suggesting it be done according to established scientific knowledge. There may be some uncertainty about how to partition data, which we discuss in Section \ref{sec:discussion}. A formal approach to testing if the data split was done appropriately can be interpreted as a test of the over-identifying restrictions: if the blocks are improperly specified (in number, size, etc.), the estimating equation $\bPsi_N$ will have mismatched moment restrictions on $\bbeta$. Formally, we can show that $Q_N$ evaluated at $\widehat{\bbeta_c}$ follows a chi-squared distribution with $(J-1)p$ degrees of freedom.
\begin{theorem}[Test of over-identifying restrictions]
\label{thm:test}
Let $\widehat{\bbeta_c}=\arg \min \limits_{\bbeta} Q_N(\bbeta)$. Given conditions \ref{cond-weight-matrix}, \ref{cond-consist} and \ref{cond-norm}, $Q_N(\widehat{\bbeta_c}) \stackrel{d}{\rightarrow} \chi^2_{(J-1)p}$ as $N \rightarrow \infty$.
\end{theorem}
The proof is given in the Supplemental Material. DIMM has the advantage of an objective function that allows for formal testing, whereas GEE model selection relies on information criteria that can be subjective. The test can also be thought of as a test of the homogeneity assumption on the mean parameter $\bbeta$, since the model $g(\bmu_i)=\bX_i \bbeta$ translates into moment restrictions on $\bbeta$. Unfortunately, it may be difficult to tell if invalid moment restrictions stem from an inappropriate data split or incorrect model specification. Residual analysis for model diagnostics can remove doubt in the latter case. 

\section{IMPLEMENTATION: the parallelized one-step estimator}
\label{sec:implementation}

In practice, searching for the integrated solution of the GMM equation \eqref{def:combined-estimator} can be very slow. Iterative methods must repeatedly evaluate $\bPsi_N(\bbeta; \by)$ at each candidate $\bbeta$, which requires the computation of the pairwise CL from each block at every iteration. We propose a meta-estimator of $\bbeta$ that delivers a one-step update via a linear function of MCLE's $\widehat{\bbeta_j}$. Our derivation of the one-step estimator is rooted in asymptotic properties of the estimating equations $\bPsi_j$ and $\bPsi_N$ in a neighbourhood of $\bbeta_0$, in a similar spirit to Newton-Raphson. Let $\bs_{j, \bpsi_j}(\bbeta)=-\nabla_{\bbeta} E_{\bbeta} \bpsi_j (\bbeta; \by_{i,j}, \widehat{\bgamma_j})$. Since $\bPsi_j (\widehat{\bbeta_j} ; \by_j, \widehat{\bgamma_j})=0$, under some regularity conditions we can approximate $\bPsi_j(\widehat{\bbeta_c}; \by_j, \widehat{\bgamma_j}) \approx \bPsi_j (\widehat{\bbeta_j}; \by_j, \widehat{\bgamma_j}) - \bs_{j, \bpsi_j}(\widehat{\bbeta_j}) (\widehat{\bbeta_c} - \widehat{\bbeta_j}) =- \bs_{j, \bpsi_j}(\widehat{\bbeta_j}) (\widehat{\bbeta_c} - \widehat{\bbeta_j})$, implying
\begin{align*}
\bPsi_N (\widehat{\bbeta_c}; \by) &\approx \left( \begin{array}{c}
\bs_{1, \bpsi_1} (\widehat{\bbeta_1})(\widehat{\bbeta_c} - \widehat{\bbeta_1}) \\
\vdots \\
\bs_{J, \bpsi_J}(\widehat{\bbeta_J}) (\widehat{\bbeta_c} - \widehat{\bbeta_J})
\end{array} \right)
=-\mbox{diag} \left\{ \bs_j (\widehat{\bbeta_j}) \right\}_{j=1}^J \left( \begin{array}{c}
\widehat{\bbeta_c} - \widehat{\bbeta_1} \\
\vdots \\
\widehat{\bbeta_c} - \widehat{\bbeta_J}
\end{array} \right).
\end{align*}
Define $\bs_{\bpsi}(\bbeta) = -\nabla_{\bbeta} E_{\bbeta} \bpsi_N(\bbeta; \by_i)$, so that $\bs_{\bpsi}(\widehat{\bbeta_c})=-\nabla_{\bbeta} E_{\bbeta,} \bpsi_N(\bbeta; \by_i) \rvert_{\bbeta = \widehat{\bbeta_c}}$ is $Jp \times p$ dimensional. As the minimizer of $Q_N(\bbeta)$, $\widehat{\bbeta_c}$ satisfies
\begin{align}
0 &= \bs^T_{\bpsi} (\widehat{\bbeta_c}) \widehat{\bV}^{-1}_{N, \bpsi} \bPsi_N(\widehat{\bbeta_c}; \by) \approx \bs^T_{\bpsi}(\widehat{\bbeta_c}) \widehat{\bV}^{-1}_{N, \bpsi} \mbox{diag} \left\{ \bs_j (\widehat{\bbeta_j}) \right\}_{j=1}^J \left( \begin{array}{c}
\widehat{\bbeta_c} - \widehat{\bbeta_1} \\
\vdots \\
\widehat{\bbeta_c} - \widehat{\bbeta_J}
\end{array} \right).
\label{combined-EE-approx}
\end{align}
Let $\left[ \widehat{\bV}^{-1}_{N, \bpsi} \right]_{i,j}$ be the rows $(i-1)p+1$ to $ip$ and columns $(j-1)p+1$ to $jp$ of matrix $\widehat{\bV}^{-1}_{N, \bpsi}$. Solving \eqref{combined-EE-approx} for $\widehat{\bbeta_c}$ leads to a closed-form expression given by:
\begin{align*}
\widehat{\bbeta_c}&\approx \left\{ \sum \limits_{i,j=1}^J \bs_{i, \bpsi_i}(\widehat{\bbeta_c}) \left[ \widehat{\bV}^{-1}_{N, \bpsi} \right]_{i,j} \bs_{j, \bpsi_j}(\widehat{\bbeta_j}) \right\}^{-1} \sum \limits_{i,j=1}^J \bs_{i, \bpsi_i}(\widehat{\bbeta_c}) \left[ \widehat{\bV}^{-1}_{N, \bpsi} \right]_{i,j} \bs_{j, \bpsi_j} (\widehat{\bbeta_j}) \widehat{\bbeta_j}.
\end{align*}
Let $\bS_{j, \bpsi_j}(\bbeta; \by_j)$ be a $\sqrt{N}$-consistent sample estimate of $\bs_{j, \bpsi_j}(\bbeta)$. Since $\widehat{\bbeta_j}$ and $\widehat{\bbeta_c}$ are both consistent for $\bbeta$, we substitute $\widehat{\bbeta_j}$ for $\widehat{\bbeta_c}$ to obtain a one-step estimator of $\bbeta$:
\begin{equation}
\resizebox{.941 \textwidth}{!} 
{
    $\widehat{\bbeta}_{DIMM}=\left( \sum \limits_{i,j=1}^J \bS_{i, \bpsi_i} (\widehat{\bbeta_i}; \by_i) \left[ \widehat{\bV}^{-1}_{N, \bpsi} \right]_{i,j} \bS_{j, \bpsi_j} (\widehat{\bbeta_j}; \by_j) \right)^{-1} \sum \limits_{i,j=1}^J \bS_{i, \bpsi_i} (\widehat{\bbeta_i}; \by_i) \left[ \widehat{\bV}^{-1}_{N, \bpsi} \right]_{i,j} \bS_{j, \bpsi_j} (\widehat{\bbeta_j}; \by_j) \widehat{\bbeta_j}$.
\label{one-step-estimator}
}
\end{equation}
With $\widehat{\bbeta}_{DIMM}$ in \eqref{one-step-estimator}, DIMM can be implemented in a fully parallelized and scalable computational scheme following, for example, the MapReduce paradigm on the Hadoop platform, where only one pass through each block of data is required. These passes can be run on parallel CPUs, and return values of summary statistics $\left( \widehat{\bbeta_j}, \bpsi_j(\widehat{\bbeta_j}; \by_{i,j}, \widehat{\bgamma_j}), \bS_{j, \bpsi_j}(\widehat{\bbeta_j}; \by_j) \right)_{j=1}^J$. After computing $\widehat{\bV}_{N, \bpsi}$ as a function of these summary statistics, computation of $\widehat{\bbeta}_{DIMM}$ in \eqref{one-step-estimator} can be done in one step. Big data stored on several servers never need be combined, meaning DIMM can be run on distributed correlated response data. $\widehat{\bbeta}_{DIMM}$ can also be used for sub-group analyses, as in Section \ref{sec:data}, to combine estimates from specific sub-groups of interest. To our knowledge, DIMM is the first method able to analyze high-dimensional distributed correlated data in a computationally fast and statistically efficient way.\\
We show in Theorem \ref{thm:dist-equiv} that the one-step estimator $\widehat{\bbeta}_{DIMM}$ in \eqref{one-step-estimator} has the same asymptotic distribution as and is asymptotically equivalent to $\widehat{\bbeta_c}$.
\begin{theorem}
\label{thm:dist-equiv}
Given conditions \ref{cond-weight-matrix}, \ref{cond-consist}, \ref{cond-norm} and \ref{cond-equiv}, $\widehat{\bbeta}_{DIMM}$ and $\widehat{\bbeta_c}$ have the same asymptotic distribution: $\sqrt{N} \left( \widehat{\bbeta}_{DIMM}-\bbeta_0 \right) \stackrel{d}{\rightarrow} \mathcal{N} \left(0, \bj_{\bpsi}^{-1}(\bbeta_0) \right)$ as $N \rightarrow \infty$. Moreover, $\widehat{\bbeta_c}$ and $\widehat{\bbeta}_{DIMM}$ are asymptotically equivalent: $\left\| \widehat{\bbeta}_{DIMM}-\widehat{\bbeta_c} \right\| \stackrel{p}{\rightarrow} 0$ as $N \rightarrow \infty$.
\end{theorem}
The proof of this theorem is given in the Supplemental Material. The additional conditions specify the convergence rate of the MCLE's $\widehat{\bbeta_j}$ to ensure the proper convergence rate of $\widehat{\bbeta}_{DIMM}$. These are necessary because the computation of the one-step estimator relies solely on the MCLE's. This theorem is the key result that allows us to use the one-step estimator, which is more computationally attractive than $\widehat{\bbeta_c}$, without sacrificing any of the asymptotic properties enjoyed by $\widehat{\bbeta_c}$, such as estimation efficiency. \\
Finally, it is clear from Theorem \ref{thm:dist-equiv} and the form of the Godambe information $\bj_{\bpsi}(\bbeta)=\sum \limits_{i,j=1}^J \bs_{i, \bpsi_i}(\bbeta) \left[ \bv^{-1}_{\bpsi} (\bbeta) \right]_{i,j} \bs_{j, \bpsi_j}(\bbeta)$ that under conditions \ref{cond-weight-matrix}-\ref{cond-equiv}, a consistent estimator of the asymptotic covariance of $\widehat{\bbeta}_{DIMM}$ is
\begin{align*}
\left( N \sum \limits_{i,j=1}^J \bS_{i, \bpsi_i} (\widehat{\bbeta_i}; \by_i) \left[ \widehat{\bV}^{-1}_{N, \bpsi} \right]_{i,j} \bS_{j, \bpsi_j} (\widehat{\bbeta_j}; \by_j) \right)^{-1}.
\end{align*}
Equipped with $\widehat{\bbeta}_{DIMM}$ and an estimate of its asymptotic covariance, we can do Wald tests and construct confidence intervals for inference on $\bbeta$. When conditions \ref{cond-weight-matrix}-\ref{cond-equiv} hold, it is clear that $Q_N(\widehat{\bbeta}_{DIMM})\stackrel{d}{\rightarrow} \chi^2_{(J-1)p}$ as $N \rightarrow \infty$, allowing us to test the goodness-of-fit of our model. When $Jp$ grows very large ($\approx 5000$), inversion of $\widehat{\bV}_{N, \bpsi}$ can be numerically unstable, although we have never encountered such a situation. In this case, we propose using a regularized modified Cholesky decomposition of $\widehat{\bV}_{N, \bpsi}$ following \cite{Pourahmadi}. Computation of a regularized estimate of $\widehat{\bV}_{N, \bpsi}^{-1}$ requires the inversion of a diagonal matrix, which is fast to compute, and the selection of a tuning parameter by cross-validation. Details are available in the Supplemental Material, and our R package allows for the implementation of a regularized weight matrix. \\
In summary, DIMM proceeds in three steps:
\begin{enumerate}[label={Step \arabic*},leftmargin=8eX]
\item Split the data according to established scientific knowledge to form $J$ blocks of lower-dimensional response subvectors with homogeneous correlations.
\item Analyze the $J$ blocks in parallel using pairwise CL. MCLE's are obtained using the R function \verb|optim|. We run 500 iterations of Nelder-Mead with initial values $\bbeta=(1, \ldots, 1)^T$. End values of this optimization are used as starting values for the BFGS algorithm, which yields $\widehat{\bbeta_j}$. We return $\left\{ \widehat{\bbeta_j}, \bpsi_j( \widehat{\bbeta_j} ; \by_{i,j}, \widehat{\bgamma_j} ), \bS_{j, \bpsi_j} (\widehat{\bbeta_j}; \by_j) \right\}_{j=1}^J$. 
\item Compute $\widehat{\bV}_{N, \bpsi}$ and then find $\widehat{\bbeta}_{DIMM}$ in \eqref{one-step-estimator}.
\end{enumerate}
An R package to implement DIMM is provided in the Supplemental Material and will be submitted to the Comprehensive R Archive Network (CRAN) shortly.

\section{SIMULATIONS}
\label{sec:simulations}

We examine through simulations the performance and finite sample properties in Theorem \ref{thm:dist-equiv} of the one-step estimator $\widehat{\bbeta}_{DIMM}$ under the linear regression setting $\bmu_i=\bX_i \bbeta$, where $\bmu_i=E(\bY_i \lvert \bX_i, \bbeta)$, $\bY_i \sim \mathcal{N} (\bX_i \bbeta, \bSigma)$. We consider two sets of simulations: the first illustrates DIMM for different dimensions $M$ of $\bY$, $J=5$ for all settings, with an intercept included in $\bX_i$, and varying number of covariates; the second pushes DIMM to its extremes with very large $M$ and $J$, and five covariates. In both settings, to mimic the infant EEG data, we let $\bSigma=\bS \otimes \bA$ with nested correlation structure, where $\otimes$ denotes the Kronecker product, $\bA$ an AR(1) covariance matrix, and $\bS$ a $J \times J$ positive-definite matrix.\\
$\left\{ \bY_i, \bX_i \right\}_{i=1}^N$ can be partitioned into $J$ blocks of data with AR(1) covariance within each block. Data within each block is modelled using the bivariate normal marginal distribution. We note that $\widehat{\bbeta_j}$ has a closed-form solution following generalized least squares (GLS): estimating $\widehat{\bbeta_j}$ can be done by iteratively updating $\widehat{\bbeta_j}^{(k)}=( \bX_j^T \widehat{\bSigma}_j^{(k)} \bX_j )^{-1} \bX_j^T \allowbreak \left\{ \widehat{\bSigma}_j^{(k)}\right\}^{-1}  \by_j$ and $\widehat{\bSigma}_j^{(k)}$, where $\widehat{\bSigma}_j^{(k)}$ has a certain known covariance structure, for $k=1,2,\ldots$ until convergence. We choose to use \verb|optim| instead of GLS because it allows our functions to be more general and performs as well as GLS; for comparison, results using GLS can be found in the Supplemental Material. True value of $\bbeta$ is set to $\bbeta_0=(0.3, 0.6, 0.8, 1.2, 0.45, 1.6)^T$ in the case of five covariates, and subsets thereof for fewer covariates.\\
We discuss the first set of simulations. Let sample size be $N=1,000$ and the AR(1) covariance matrix $\bA$ have standard deviation $\sigma=2$ and correlation $\rho=0.5$. CL estimation of $\widehat{\bbeta_j}$ is done first by using the correct AR(1) block covariance structure (DIMM-AR(1)). To evaluate how our method performs under covariance misspecification, we estimate $\widehat{\bbeta_j}$ using a compound symmetry (DIMM-CS) block covariance structure. \\
We compute $\widehat{\bbeta}_{DIMM}$ from \eqref{one-step-estimator} and its covariance, and report root mean squared error (RMSE), empirical standard error (ESE), mean asymptotic standard error (ASE), and mean bias (BIAS) with $M=200$ and five scalar covariates (Table \ref{simulations-1}) and with $M=1,000$ and two vector covariates (Table \ref{simulations-2}). We compare DIMM to estimates of $\bbeta$ obtained using GEE with a compound symmetry covariance structure (GEE-CS) and independence covariance structure (GEE-IND) using the R library \verb|geepack| (\cite{geepack}), and using GLS with known covariance (GLS-oracle) (our code). The latter can be considered the ``oracle setting'', as we do not estimate the covariance of the response but use the true covariance to estimate $\bbeta$. We examine type-I error of the test $H_0:\beta_q=0$ for $q=1, \ldots, p$ for each simulation scenario, and the chi-squared distribution of test statistic $Q_N(\widehat{\bbeta}_{DIMM})$ with $M=200$ and one covariate (see Supplemental Material). Simulations are conducted using R software on a standard Linux cluster with 16GB of random-access memory per CPU. CL evaluation is coded in C++ but minimization of the CL occurs in R.\\
In Table \ref{simulations-1}, $\widehat{\bbeta}_{DIMM}$ appears consistent since BIAS is close to zero. RMSE, ESE and ASE are approximately equal, meaning DIMM is unbiased and has correct variance formula in Theorem \ref{thm:dist-equiv}. Moreover, DIMM mean variance is generally smaller than GEE mean variance. In data analyses, this results in increased statistical power and more signal detection. Finally, DIMM is close to attaining the estimation efficiency under the GLS-oracle case of known covariance, which is the best efficiency possible. In Table \ref{simulations-2}, we corroborate these observations for spatially/longitudinally-varying vector covariates. Our method also still performs well when dimension is equal to sample size. Finally from Figure \ref{simulations-3}, we see that DIMM is computationally much faster than GEE and GLS, and maintains appropriate confidence interval coverage, corroborating the theoretical asymptotic distribution in Theorem \ref{thm:dist-equiv} for large sample size.  For fixed $J$, DIMM is scalable, since the dimension of $\bPsi_N$ does not increase. We remark that CPU time consists of time spent by the CPU on calculations and is generally shorter than elapsed time. Elapsed time depends greatly on the implementation and hardware, and is therefore harder to compare between methods.
\begin{table}[H]
\centering
\ra{1.2}
\caption{Simulation results: RMSE, BIAS, ESE, ASE with five covariates, $N=1,000$, $M=200$, $J=5$, averaged over 500 simulations.}
\label{simulations-1}
\resizebox{\columnwidth}{!}{%
\begin{tabular}{@{}rrrrrrr@{}}
\toprule
\multicolumn{1}{c}{} & measure$\times 10^{-2}$ & DIMM-AR(1) & DIMM-CS & GEE-CS & GEE-IND &GLS-oracle \\ 
\midrule
$\beta_0$ & RMSE/BIAS & 4.323/$-0.315$ & 4.324/$-0.317$ & 4.881/$-0.334$ & 4.881/$-0.334$  & 4.117/$-0.357$ \\ 
& ESE/ASE & 4.316/4.213 & 4.317/4.214 & 4.874/4.848 & 4.874/4.848 & 4.106/4.123 \\ 
$\beta_1$ & RMSE/BIAS & 1.837/0.043 & 1.837/0.043 & 2.092/0.08 & 2.092/0.08 & 1.796/0.062 \\ 
& ESE/ASE & 1.838/1.779 & 1.838/1.78 & 2.092/2.052 & 2.092/2.052 & 1.797/1.739 \\ 
$\beta_2$ & RMSE/BIAS & 3.469/$-0.065$ & 3.471/$-0.066$ & 3.753/0.084 & 3.753/0.084 & 3.243/$-0.021$ \\ 
& ESE/ASE & 3.472/3.23 & 3.474/3.232 & 3.756/3.725 & 3.756/3.725 & 3.247/3.168  \\ 
$\beta_3$ & RMSE/BIAS & 1.509/0.138 & 1.511/0.139 & 1.671/0.095 & 1.671/0.095 & 1.453/0.132 \\ 
& ESE/ASE & 1.505/1.45 & 1.506/1.45 & 1.67/1.669 & 1.67/1.669 & 1.449/1.419 \\ 
$\beta_4$ & RMSE/BIAS & 5.486/0.196 & 5.485/0.199 & 5.982/0.193 & 5.982/0.193 & 5.256/0.288 \\ 
& ESE/ASE & 5.488/5.152 & 5.487/5.153 & 5.984/5.923 & 5.984/5.923 & 5.254/5.036 \\ 
$\beta_5$ & RMSE/BIAS & 3.56/$-0.072$ & 3.56/$-0.072$ & 3.988/$-0.081$ & 3.988/$-0.081$ & 3.424/$-0.044$ \\ 
& ESE/ASE & 3.563/3.208 & 3.563/3.209 & 3.991/3.741 & 3.991/3.741 & 3.427/3.175 \\ 
\bottomrule
\end{tabular}%
}
\parbox{\textwidth}{\small
\vspace{1eX} 
Response block sizes are $(m_1, m_2, m_3, m_4, m_5)=(45, 42, 50, 34, 29)$. $X_1 \sim Normal(0,1)$, $X_2\sim Bernoulli(0.3)$, $X_3 \sim Categorical(0.1, 0.2, 0.4, 0.25, 0.05)$, $X_4 \sim Uniform(0,1)$, and $X_5=X_1 \times X_2$ an interaction term.}
\end{table}
\begin{table}[h!]
\centering
\ra{1.2}
\caption{Simulation results: RMSE, BIAS, ESE, ASE with two covariates, $N=1,000$, $M=1,000$, $J=5$, averaged over 500 simulations.}
\label{simulations-2}
\resizebox{\columnwidth}{!}{%
\begin{tabular}{@{}rrrrrrr@{}}
\toprule
\multicolumn{1}{c}{} & measure$\times 10^{-2}$ & DIMM-AR(1) & DIMM-CS & GEE-CS & GEE-IND &GLS-oracle \\ 
\midrule
$\beta_0$ & RMSE/BIAS & 0.718/0.009 & 0.716/0.008 & 0.824/0.010 & 0.824/0.010 & 0.693/$-0.002$ \\
& ESE/ASE & 0.718/0.716 & 0.716/0.716 & 0.825/0.816 & 0.825/0.816 & 0.694/0.697 \\
$\beta_1$ & RMSE/BIAS & 0.196/$-0.002$ & 0.196/0.000 & 0.205/0.003 & 0.205/0.002 & 0.126/$-0.002$ \\
& ESE/ASE & 0.197/0.190 & 0.196/0.191 & 0.205/0.201 & 0.205/0.201 & 0.126/0.126 \\
$\beta_2$ & RMSE/BIAS & 0.457/0.012 & 0.456/0.014 & 0.516/0.001 & 0.516/0.001 & 0.438/0.023 \\
& ESE/ASE & 0.457/0.458 & 0.457/0.458 & 0.517/0.518 & 0.517/0.518 & 0.437/0.445 \\
\bottomrule
\end{tabular}%
}
\parbox{\textwidth}{\small
\vspace{1eX} Response block sizes are $(m_1, m_2, m_3, m_4, m_5)=(225, 209, 247, 170, 149)$. $X_1 \sim Normal_M(0,S)$, where $S$ is a positive-definite $M \times M$ matrix, $X_2$ a vector of alternating $0$'s and $1$'s to imitate an exposure.}
\end{table}
We now discuss the second set of simulations. We let sample size $N=1,500$ and consider a very challenging linear regression problem with high-dimension $M=10,000$, and $J=12$ such that $(m_1, \ldots, m_{12})=(917, 863, 988, 734, 906, 603, 756, 963, \allowbreak 915, 856, 641, 858)$. We let $\bX_{i}$ be a matrix of five covariates and an intercept, and the AR(1) covariance matrix $\bA$ with standard deviation $\sigma=16$ and correlation $\rho=0.8$. We compute $\widehat{\bbeta}_{DIMM}$ from \eqref{one-step-estimator} and its estimated covariance, and plot RMSE, ESE, ASE, and BIAS in Figure \ref{simulations-4}. We were unable to compare DIMM with existing competitors due to the tremendous computational burden associated with such high-dimensional $M$. As in the first set of simulations, $\widehat{\bbeta}_{DIMM}$ is consistent with ignorable BIAS. RMSE, ESE and ASE are approximately equal, confirming the large-sample properties of DIMM in this numerical example. We remark that ASE slightly underestimates ESE for certain covariate types. This could be due to the large dimensionality $Jp=72$ of $\bPsi_N$ and/or the poorer performance of GMM in smaller samples (see Section \ref{sec:discussion}).
\begin{figure}[H]
\begin{center}
\includegraphics[width=\linewidth]{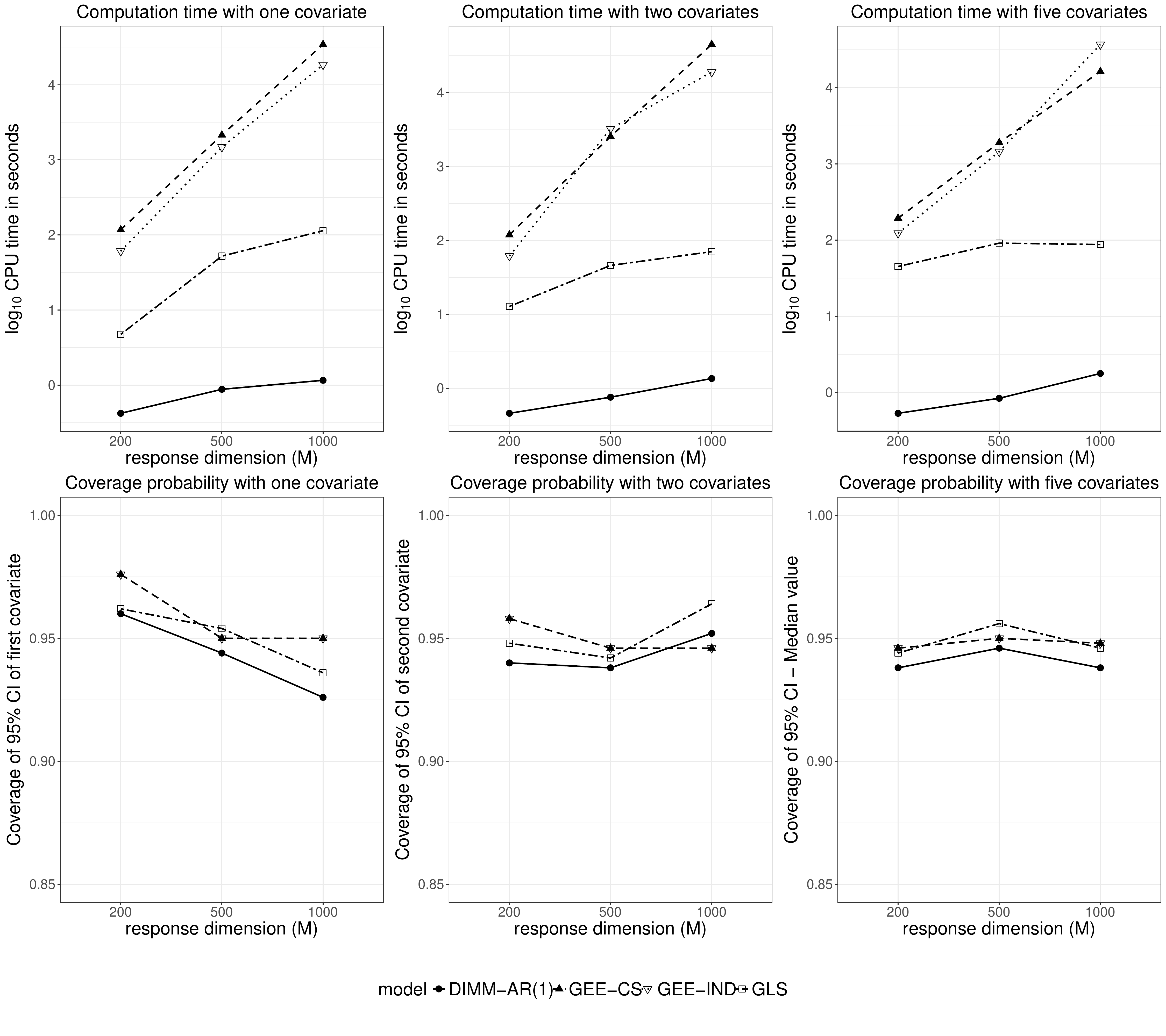}
\end{center}
\vspace{-0.5cm}\parbox{\textwidth}{\caption{\small Upper panels: comparison of computation time on $\log_{10}$ scale of four methods for varying dimension $M$ based on 500 simulations. Lower panels: comparison of 95\% confidence interval coverage of four methods for varying dimension $M$ based on 500 simulations. Left column has $X_1 \sim \mathcal{N}(0,1)$; middle column has $X_1 \sim \mathcal{N}_M(0,S)$, where $S$ is a positive-definite $M \times M$ matrix, and $X_2$ a vector of alternating 0's and 1's; right column has $X_1 \sim \mathcal{N}(0,1)$, $X_2 \sim Bernoulli(0.3)$, $X_3 \sim Multinomial(0.1, 0.2, 0.4, 0.25, 0.05)$, $X_4 \sim Uniform(0,1)$, and $X_5$ an interaction between $X_1$ and $X_2$. \label{simulations-3}}}
\end{figure}
\begin{figure}[H]
\begin{center}
\includegraphics[width=0.9\linewidth]{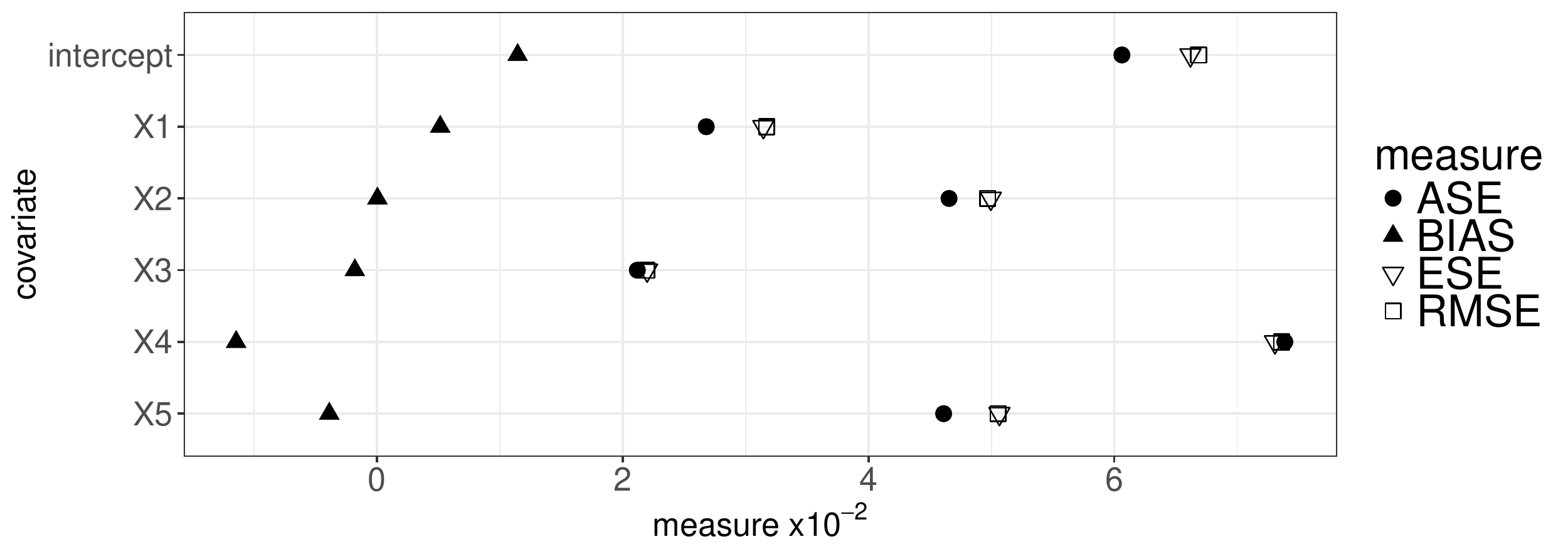}
\end{center}
\vspace{-0.5cm}\parbox{\textwidth}{\caption{\small RMSE, BIAS, ESE, ASE based on 100 simulations with an intercept and five covariates, and $M=10,000$. Covariates are simulated as in the right column of Figure \ref{simulations-3}. \label{simulations-4}}}
\end{figure}
\vspace{-0.5cm}
Beyond theoretical validation, the simulation results presented in this section highlight the applicability, flexibility and computational power of DIMM. The empirical evidence from simulations is encouraging and advocates the ability of DIMM to deal with high-dimensional correlated response data with multi-level nested correlations.

\section{APPLICATION TO INFANT EEG DATA}
\label{sec:data}

\begin{figure}[h!]
\begin{center}
\centerline{\includegraphics[width=0.81\linewidth]{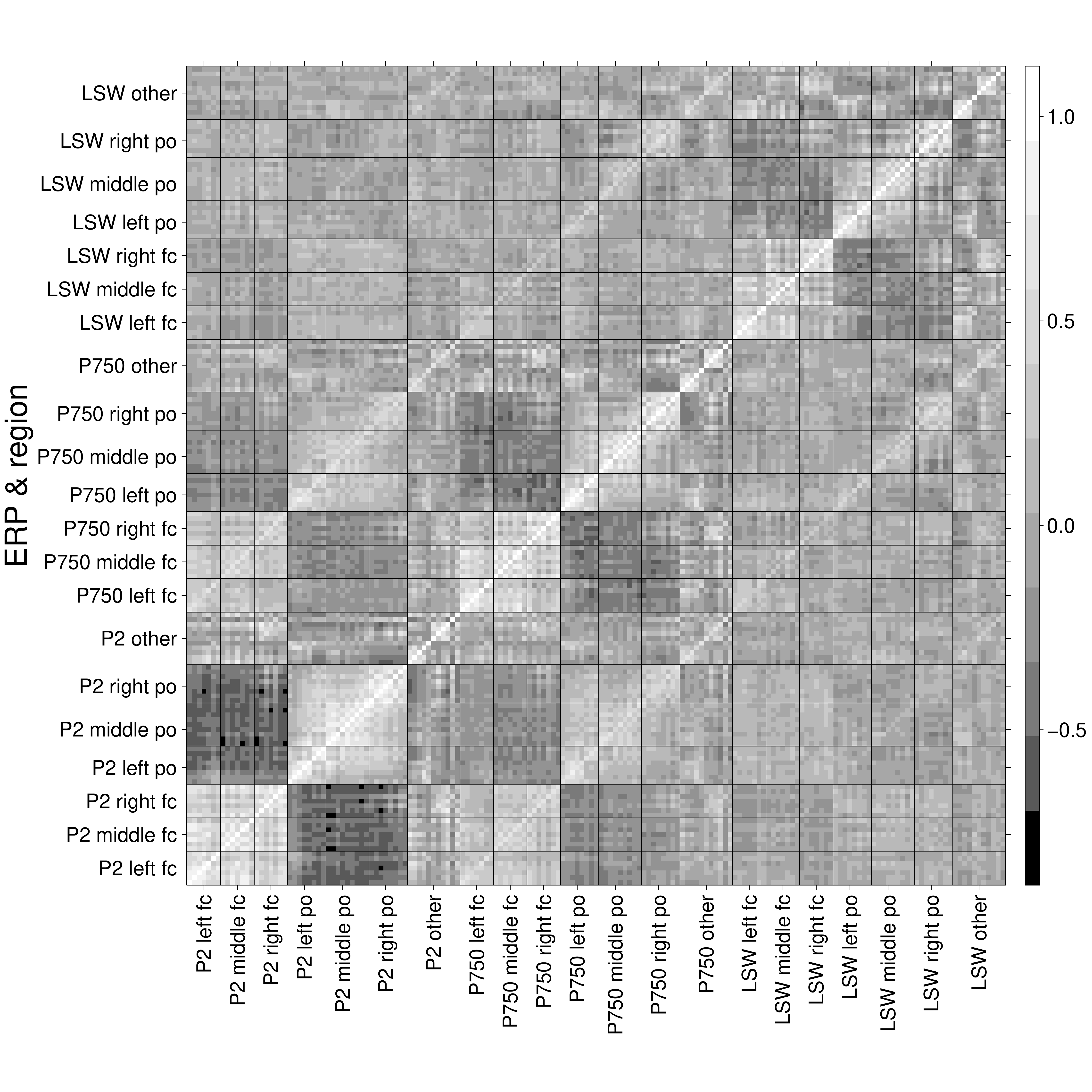}}
\end{center}
\vspace{-1.7cm}\caption{\small Correlation of electrical amplitude at three ERP's for iron sufficient children under stimulus of mother's voice (color plot and additional plots in Supplemental Material).}\label{heatmap-correlation}
\end{figure}
We present the analysis of the infant EEG data introduced in Section \ref{sec:intro}. EEG data from 157 two-month-old infants under two stimuli at 46 nodes was used. Six brain regions were identified by the investigator as related to auditory recognition memory, with an additional reference node (VREF), as visualized in Figure \ref{map-nodes-regions}: left frontal-central (11, 12, 13, 14, 15, 18, 19), middle frontal-central (3, 4, 6, 7, 8, 9, 54), right frontal-central (2, 53, 56, 57, 58, 59, 60), left parietal-occipital (24, 25, 26, 27, 28, 29, 30, 32), middle parietal-occipital (31, 33, 34, 35, 36, 37, 38, 39, 40),  and right parietal-occipital (42, 43, 44, 45 46, 47, 48, 52). \\
The primary scientific objective of this study is to quantify the effect of iron deficiency on auditory recognition memory. From cord blood at birth, 50 infants were classified as iron deficient ($sufficiency\_ \allowbreak status$ $= 1$) and 107 as iron sufficient based on serum ferritin and zinc protoporphyrin levels. Additional available covariates are age and type of stimulus (mother's voice coded with $voice\_stimulus=1$). The response for one infant has a complex nested correlation structure; see Figure \ref{heatmap-correlation}. This figure aligns with substantive scientific knowledge and suggests a partition of data into $18$ blocks of response subvectors, one for each ERP and brain region. It also corroborates prior knowledge of high correlations within frontal-central regions, parietal-occipital regions, and between ERPs P2 and P750.\\
Let $\bY_{i,j}$ be the vector of EEG measurements in one brain region and ERP (block $j$, $j=1, \ldots, 18$) for infant $i$, and consider the linear model  with block-specific coefficients:
\begin{align}
E\left(\bY_{i,j}\right)&=\beta_{0,j}+\beta_{1,j}age_{i,j} +\beta_{2,j} voice\_stimulus_{i,j}+\beta_{3,j} sufficiency\_status_{i,j}.  \label{homog}
\end{align}
Instead of assuming global homogeneous covariate effects, which is not biologically meaningful, we perform analyses based on certain locally homogeneous covariate-response relationships to identify specific regions affected or not by iron deficiency. Through individual block analyses (see Supplemental Material) and existing knowledge, we identify homogeneous covariate effects across frontal-central regions in each ERP, the left parietal-occipital region in P2 and P750, the middle and right parietal-occipital regions from P2, the middle and right parietal-occipital regions from P750, and parietal-occipital regions from LSW. As mentioned previously, DIMM's flexibility allows us to conduct sub-group analyses by combining blocks of homogeneous effects to improve statistical power.\\
We use an inverse normal transformation of the responses for each analysis. To estimate regression parameters using DIMM, we assume a compound symmetric covariance structure of the response within each brain region and each ERP; block analyses are run in parallel; we compute the one-step estimator $\widehat{\bbeta}_{DIMM}$ for the set of homogeneous regions of interest. We compare DIMM to GEE-CS to reinforce gains in computation time and statistical power. Based on simulations mimicking our data setting (see Supplemental Material), we find that DIMM and GEE-CS have adequate power despite limited sample size. We present iron sufficiency status effect estimates for selected sub-group analyses in Table \ref{data-results} (complete results available in the Supplemental Material).\\
\begin{table}[h!]
\centering
\ra{1.2}
\caption{Select EEG data analysis results: iron sufficiency status effect estimates and statistics for each combination scheme.}
\begin{tabular}{@{}lrrrrr@{}}
\toprule
\multirow{2}{*}{combine region, ERP} & \multirow{2}{*}{method} & estimate & \multirow{2}{*}{p-value} & CPU & CPU\\
& & (s.d.$\times 10^{-2}$) & & seconds & time ratio \\
\midrule
left, middle and right fc, P2 & GEE-CS & 0.103 (12.0) & 0.391 & 0.701 & \multirow{2}{*}{0.47}\\
& DIMM & 0.087 (11.9) & 0.468 & 1.497 & \\
left po, P2 \& P750 & GEE-CS & -0.174 (8.3) & 0.036 & 0.209 & \multirow{2}{*}{1.20}\\
& DIMM & -0.226 (8.1) & 0.005 & 0.174 & \\
left, middle and right po, LSW & GEE-CS & 0.041 (8.7) & 0.640 & 0.528 & \multirow{2}{*}{2.90}\\
& DIMM & 0.087 (8.4) & 0.298 & 0.182 & \\
\bottomrule
\end{tabular}
\parbox{\textwidth}{\small
\vspace{1eX} fc, frontal-central; po, parietal-occipital; s.d., standard deviation.}
\label{data-results}
\end{table}
For all analyses, DIMM finds a more precise estimate than GEE. This is because the compound symmetry covariance structure assumed by GEE over the entire response may not be close to the true covariance, resulting in a loss of efficiency. For all but one analysis -- left, middle and right frontal-central, P2 -- DIMM also performs faster than GEE. This is because of the parallelization of DIMM. DIMM may be slower than GEE in the one analysis because of the limited sample size and small response dimensionality, limiting the improvements of DIMM over GEE. Nonetheless, in data simulations (see Supplemental Material), on average DIMM performs faster than GEE for this analysis. Effect estimates from GEE and DIMM tend to be in the same direction, increasing confidence in our results. We find a significant result for the left parietal-occipital region in P2 \& P750. Iron deficient infants had expected transformed left parietal-occipital P2 \& P750 amplitude 0.226 units lower than iron sufficient infants of the same age and sex. By making better model assumptions and running block analyses in parallels, we find more precise estimates of iron sufficiency status effect faster than using GEE's. The proposed DIMM shows promise in simple data analyses, and has the theoretical justification to perform well in more complex scenarios.

\section{DISCUSSION}
\label{sec:discussion}

The proposed DIMM allows for the fast and efficient estimation of regression parameters with high-dimensional correlated response. Simulations show the scalability of DIMM for fixed $J$ and confirm key asymptotic properties of the DIMM estimator.\\
The $\widehat{\bbeta}_{DIMM}$ estimator can be implemented using a fully parallelized computational scheme, for example using the MapReduce paradigm on the Hadoop platform. Investigators split data into blocks of responses with simple and homogeneous covariance structures. The data partition may be driven by some established scientific knowledge or certain data-driven approaches. Errors in prior knowledge can lead to misspecification of the data split, which may be checked via model diagnostics or goodness-of-fit test statistics. If sample size is large enough, investigators may consider imposing no or very little structure on $\bgamma_j$ to avoid misspecifying response blocks.\\
Potential trade-offs between number of blocks $J$ and block size $m_j$ should be evaluated when there is no strong substantive knowledge to guide the choice of partition. Our numerical experience has suggested that although large $J$ leads to smaller $m_j$ and therefore faster computation and less strict model assumptions, DIMM may yield inefficient results due to large dimensionality of the integrated CL score vector $\bPsi_N$. On the other hand, large $m_j$ but small $J$ will have the opposite effect of slower computation and stricter model assumptions within each block but better combination of results.\\
Finally, issues related to poor performance of GMM in small samples have been documented in the literature and must be considered when sample size is small; for a discussion, see \cite{Hansen-Heaton-Yaron} and others in the special section on small-sample properties of GMM in the Journal of Business and Economic Statistics. In this case, to reduce the dimensionality of the integrated CL score vector $\bPsi_N$, we suggest integrating analyses from a small number of blocks for more reliable results, as done in Section \ref{sec:data}.\\
DIMM utilizes the full strength of GMM to combine information from multiple sources to achieve greater statistical power, an approach that has been shown to work well with longitudinal data; see for examples \cite{Wang-Wang-Song-2012} and \cite{Wang-Wang-Song-2016}. DIMM has the potential to combine multimodal data, an important analytic task in biomedical data analysis for personalized medicine. Indeed, response data in each block can be modelled using any pairwise distribution $f_j$, where $\left\{ f_j\right\}_{j=1}^J$ can be made compatible with $f(\bY; \bGamma)$ using Fr\'{e}chet classes (see Chapter 3 of \cite{Joe-1}). We anticipate numerous extensions to DIMM, including the addition of penalty terms to CL estimating equations, and allowing for spatially varying mean parameter $\bbeta$ and prediction of neighbouring response variables. We anticipate that DIMM will be useful for many types of data, including genomic, epigenomic, and metabolomic, indicating the promising methodological potential of DIMM.\\

\appendix

\section{APPENDIX: proofs of asymptotic properties}

\label{sec:appendix:1}

Let $\Theta$ be the compact parametric space of $\bbeta$ and let $\left\| \cdot \right\|$ be the Euclidean norm. We list the regularity conditions required to establish large samples properties in the paper.
\begin{enumerate}[label=C.\arabic*]
\item \label{cond-weight-matrix} Assume $E_{\bbeta_0} \bPsi_N (\bbeta; \by)$ has a unique zero at $\bbeta_0$, $-\nabla_{\bbeta} E_{\bbeta} \bpsi_N(\bbeta; \by_i)$ is smooth in a neighbourhood $\mathcal{N}(\bbeta_0)$ of $\bbeta_0$ and positive definite, $\bv_{\bpsi}(\beta_0)$ is finite, positive-definite and nonsingular, and $\left\| \bpsi_N(\bbeta_1; \by_i) - \bpsi_N(\bbeta_2; \by_i) \right\| \leq C \left\| \bbeta_1 - \bbeta_2 \right\|$ for all $\bbeta_1, \bbeta_2 \in \mathcal{N}(\bbeta_0)$ and some constant $C>0$.
\item \label{cond-consist}
Following \cite{Newey-McFadden}, define $Q_0(\bbeta)=E_{\bbeta} \left\{ \bPsi_N^T(\bbeta; \bY) \right) \bv_{\bpsi}^{-1}(\bbeta_0) \allowbreak E_{\bbeta} \left( \bPsi_N(\bbeta; \bY) \right\}$. Assume $Q_0(\bbeta)$ is twice-continuously differentiable in $\mathcal{N}(\bbeta_0)$.
\item \label{cond-norm} Let $\widehat{\bbeta_c}$ be as defined in \eqref{def:combined-estimator}, and $\bbeta_0$ an interior point of its parameter space $\Theta$. Following \cite{Newey-McFadden}, assume $Q_N(\widehat{\bbeta_c}) \leq \inf \limits_{\bbeta \in \Theta} Q_N(\bbeta) + o_p(1)$, and, for any $\delta_N \rightarrow 0$,
\begin{equation*}
\sup \limits_{\left\| \bbeta - \bbeta_0 \right\| \leq \delta_N} \frac{\sqrt{N}}{1+\sqrt{N} \left\| \bbeta - \bbeta_0 \right\| } \left\| \bPsi_N(\bbeta; \by) - \bPsi_N(\bbeta_0; \by) - E_{\bbeta} \bPsi_N(\bbeta; \bY) \right\| \stackrel{p}{\rightarrow} 0.
\end{equation*}
\item \label{cond-equiv}
For each $j=1, \ldots, J$, assume $\widehat{\bbeta_j}=\bbeta_0 + O_p(N^{-1/2})$. For any $\delta_N \rightarrow 0$, assume
\begin{equation*}
\sup \limits_{\left\| \bbeta - \bbeta_0 \right\| \leq \delta_N} \frac{\sqrt{N}}{1+\sqrt{N} \left\| \bbeta - \bbeta_0 \right\| } \left\| \bPsi_N(\bbeta; \by) - \bPsi_N(\bbeta_0; \by) - E_{\bbeta} \bPsi_N(\bbeta; \bY) \right\| =O_p(N^{-1/2}).\\
\end{equation*}
\end{enumerate}

\begin{proof}[Proof of Lemma \ref{lemma:weight-matrix}:]
For notation purposes, let $\widehat{\bbeta}_{MCLE}=\left(\widehat{\bbeta_1}, \ldots, \widehat{\bbeta_J} \right)^T$ and \\$\bpsi_N(\widehat{\bbeta}_{MCLE}; \by_i)=\left( \bpsi_1(\widehat{\bbeta_1}; \by_{i,1}, \widehat{\bgamma_1})^T, \ldots, \bpsi_J(\widehat{\bbeta_J}; \by_{i,J}, \widehat{\bgamma_J})^T \right)^T$. By consistency of the MCLE due to \ref{cond-weight-matrix}, $\widehat{\bbeta_j}-\bbeta_0=o_p(1)$. Since $J$ and $p$ finite, $\left\| \widehat{\bbeta}_{MCLE} - \bbeta_0 \right\| = o_p(1)$. Then by \ref{cond-weight-matrix},
\begin{align*}
\left\| \bpsi_N(\widehat{\bbeta}_{MCLE}; \by_i) - \bpsi_N(\bbeta_0; \by_i) \right\| &\leq C \left\| \widehat{\bbeta}_{MCLE} - \bbeta_0 \right\| =o_p(1).
\end{align*}
Plugging into $\widehat{\bV}_{N, \bpsi}$, we have
\begin{align*}
\widehat{\bV}_{N, \bpsi}&=\frac{1}{N} \sum \limits_{i=1}^N \bpsi_N(\widehat{\bbeta}_{MCLE}; \by_i) \bpsi^T_N(\widehat{\bbeta}_{MCLE}; \by_i)\\
&=\frac{1}{N} \sum \limits_{i=1}^N \bpsi_N(\bbeta_0; \by_i) \bpsi^T_N(\bbeta_0; \by_i) + o_p(1) \frac{1}{N} \sum \limits_{i=1}^N \bpsi_N(\bbeta_0; \by_i) + o_p(1)\\
&=\frac{1}{N} \sum \limits_{i=1}^N \bpsi_N(\bbeta_0; \by_i) \bpsi^T_N(\bbeta_0; \by_i) + o_p(1).
\end{align*}
Note that $\frac{1}{N} \sum \limits_{i=1}^N \bpsi_N(\bbeta_0; \by_i) \bpsi^T_N(\bbeta_0; \by_i) = \bv_{\bpsi}(\bbeta_0)+o_p(1)$. Then, $\widehat{\bV}_{N, \bpsi}=\bv_{\bpsi}(\bbeta_0)+o_p(1)$.
\end{proof}

\begin{proof}[Proof of Theorem \ref{thm:consist}:]
It is sufficient to show that, by conditions \ref{cond-weight-matrix} and \ref{cond-consist}, $\frac{1}{N} Q_N(\bbeta)$ converges uniformly in probability to $Q_0(\bbeta)$.
\begin{align*}
&\left\| \frac{1}{N} Q_N(\bbeta) - Q_0(\bbeta) \right\|\\
&~~~~=\left\| \bPsi_N^T(\bbeta; \by) \widehat{\bV}^{-1}_{N, \bpsi} \bPsi_N(\bbeta; \by) - E_{\bbeta} \left( \bPsi^T_N(\bbeta; \bY) \right) \bv_{\bpsi}^{-1}(\bbeta_0) E_{\bbeta} \left( \bPsi_N(\bbeta; \bY) \right) \right\|\\
&~~~~=\left\| \bPsi_N^T(\bbeta; \by) \widehat{\bV}^{-1}_{N, \bpsi} \bPsi_N(\bbeta; \by) \right.\\
&~~~~~~~~-2E_{\bbeta} \left( \bPsi_N^T(\bbeta; \bY) \right) \widehat{\bV}^{-1}_{N, \bpsi} \bPsi_N(\bbeta; \by) + 2E_{\bbeta} \left( \bPsi_N^T(\bbeta; \bY) \right) \widehat{\bV}^{-1}_{N, \bpsi} \bPsi_N(\bbeta; \by)\\
&~~~~~~~~-2E_{\bbeta}\left( \bPsi_N^T(\bbeta; \bY) \right) \widehat{\bV}^{-1}_{N, \bpsi} E_{\bbeta} \left( \bPsi_N(\bbeta; \bY) \right) + 2E_{\bbeta} \left( \bPsi_N^T(\bbeta; \bY) \right) \widehat{\bV}^{-1}_{N, \bpsi} E_{\bbeta} \left( \bPsi_N(\bbeta; \bY) \right)\\
&~~~~~~~~\left. -E_{\bbeta} \left( \bPsi_N^T(\bbeta; \bY) \right) \bv_{\bpsi}^{-1}(\bbeta_0) E_{\bbeta} \left( \bPsi_N(\bbeta; \bY) \right) \right\|\\
&~~~~=\left\| \bPsi_N^T(\bbeta; \by) \widehat{\bV}^{-1}_{N, \bpsi} \bPsi_N(\bbeta; \by) - \bPsi_N^T(\bbeta; \by) \widehat{\bV}^{-1}_{N, \bpsi} E_{\bbeta} \left( \bPsi_N(\bbeta; \bY) \right) \right.\\
&~~~~~~~~-E_{\bbeta} \left( \bPsi_N^T(\bbeta; \bY) \right) \widehat{\bV}^{-1}_{N, \bpsi} \bPsi_N(\bbeta; \by) + E_{\bbeta} \left( \bPsi_N^T(\bbeta; \bY) \right) \widehat{\bV}^{-1}_{N, \bpsi} E_{\bbeta} \left( \bPsi_N(\bbeta; \bY) \right)\\
&~~~~~~~~+2E_{\bbeta} \left( \bPsi_N^T(\bbeta; \bY) \right) \widehat{\bV}^{-1}_{N, \bpsi} \bPsi_N(\bbeta; \by) -2E_{\bbeta} \left( \bPsi_N^T(\bbeta; \bY) \right) \widehat{\bV}^{-1}_{N, \bpsi} E_{\bbeta} \left( \bPsi_N(\bbeta; \bY) \right)\\
&~~~~~~~~\left. +E_{\bbeta} \left( \bPsi_N^T(\bbeta; \bY) \right) \widehat{\bV}^{-1}_{N, \bpsi} E_{\bbeta} \left( \bPsi_N(\bbeta; \bY) \right) - E_{\bbeta} \left( \bPsi_N^T(\bbeta; \bY) \right) \bv_{\bpsi}^{-1}(\bbeta_0) E_{\bbeta} \left( \bPsi_N(\bbeta; \bY) \right) \right\|\\
&~~~~\leq \left\| \left[ \bPsi_N(\bbeta; \by) - E_{\bbeta} \bPsi_N(\bbeta; \bY) \right]^T \widehat{\bV}^{-1}_{N, \bpsi} \left[ \bPsi_N(\bbeta; \by) - E_{\bbeta} \bPsi_N(\bbeta; \bY) \right] \right\|\\
&~~~~~~~~+2 \left\| E_{\bbeta} \left( \bPsi^T_N(\bbeta; \bY) \right) \widehat{\bV}^{-1}_{N, \bpsi} \left[ \bPsi_N(\bbeta; \by) - E_{\bbeta} \bPsi_N(\bbeta; \bY) \right] \right\|\\
&~~~~~~~~+\left\| E_{\bbeta} \left( \bPsi_N^T(\bbeta; \bY) \right) \left[ \widehat{\bV}^{-1}_{N, \bpsi} - \bv_{\bpsi}^{-1}(\bbeta_0) \right] E_{\bbeta} \left( \bPsi_N(\bbeta; \bY) \right) \right\|\\
&~~~~\leq \left\| \bPsi_N(\bbeta; \by) - E_{\bbeta} \bPsi_N(\bbeta; \bY) \right\|^2 \left\| \widehat{\bV}^{-1}_{N, \bpsi} \right\|\\
&~~~~~~~~+ 2\left\| E_{\bbeta} \bPsi_N(\bbeta; \bY) \right\| \left\| \bPsi_N(\bbeta;\by) - E_{\bbeta} \bPsi_N(\bbeta; \bY) \right\| \left\| \widehat{\bV}^{-1}_{N, \bpsi} \right\| \\
&~~~~~~~~+\left\| E_{\bbeta} \bPsi_N(\bbeta; \bY) \right\|^2 \left\| \widehat{\bV}^{-1}_{N, \bpsi} - \bv_{\bpsi}^{-1}(\bbeta_0) \right\|\\
&~~~~=O_p(N^{-1/2}) + o_p(1).
\end{align*}
It follows that $\sup \limits_{\bbeta \in \Theta} \left\| \frac{1}{N} Q_N(\bbeta) - Q_0(\bbeta) \right\| \stackrel{p}{\rightarrow} 0$ as $N \rightarrow \infty$. By Theorem 2.1 in \cite{Newey-McFadden}, the combined GMM estimator satisfies $\widehat{\bbeta_c} \stackrel{p}{\rightarrow} \bbeta_0$ as $N \rightarrow \infty$. \end{proof}

\bibliographystyle{apalike}

\bibliography{DIMM-bib-12182017}

\end{document}